\documentclass[10pt,twocolumn]{IEEEtran}
\usepackage{soul,xcolor}
\usepackage{bm}
\include{graphics}
\include{amsmath}
\usepackage[normalem]{ulem}
\usepackage{multirow,enumitem}
\usepackage{algorithm}
\usepackage{algpseudocode}
\usepackage{verbatim}
\usepackage[latin1]{inputenc}
\usepackage{amsmath}
\usepackage{amsfonts}
\usepackage{amssymb}
\usepackage{mathrsfs}
\usepackage{booktabs}
\usepackage{url}
\usepackage[pdftex]{graphicx}
\usepackage[skip=1pt,font=scriptsize]{subcaption}
\usepackage{ifthen}
\usepackage{xspace}
\usepackage{dsfont}
\usepackage{bbm}
\usepackage{cite}
\usepackage{epsfig}
\usepackage{epstopdf}
\usepackage{array}
\usepackage{multicol}
\usepackage{algpseudocode}
\usepackage{algorithm}
\usepackage{amssymb}
\usepackage{breqn}
\usepackage{esint}
\usepackage[skip=2pt,font=footnotesize]{caption}
\usepackage{multicol}
\usepackage{amsthm}
\usepackage{setspace}
\usepackage{lipsum}
\usepackage{dblfloatfix}
\theoremstyle{plain}
\newtheorem{theorem}{Theorem}
\newtheorem{lemma}{Lemma}
\theoremstyle{definition}

\theoremstyle{remark}

\usepackage{footnote}

\newcommand{\nm}[1]{\textcolor{brown}{\textbf{[NM: #1]}}}

\usepackage{cancel}
   
\title{Trajectory Optimization for Rotary-Wing UAVs in Wireless Networks with Random Requests}
\author{Matthew Bliss and Nicol\`{o} Michelusi
\thanks{Bliss and Michelusi are with the School of Electrical and Computer Engineering, Purdue University, West Lafayette, IN, USA; emails: \{blissm,michelus\}@purdue.edu.}% <-this % stops a space
\vspace{-5mm}
}

\begin{document}
\setulcolor{red}
\setul{red}{2pt}
\setstcolor{red}
\maketitle
\thispagestyle{empty}
\pagestyle{empty}
\begin{abstract}
This paper studies the trajectory optimization problem in a scenario where a single rotary-wing UAV acts as a relay of data payloads for downlink transmission requests generated randomly by two ground nodes (GNs) in a wireless network. The goal is to optimize the UAV trajectory in order to minimize the expected average communication delay to serve these random requests. It is shown that the problem can be cast as a semi-Markov decision process (SMDP), and the resulting minimization problem is solved via multi-chain policy iteration. 
The optimality of a \emph{two-scale} optimization approach is proved: the optimal trajectory in the communication phase greedily minimizes the communication delay of the current request while moving between the current start position and a target end position (inner optimization); the end positions are selected to minimize the expected average long-term delay in the SMDP (outer optimization). Numerical simulations show that the expected average delay is minimized when the UAV moves towards the geometric center of the GNs during phases in which it is not actively servicing transmission requests, and demonstrate significant improvements over sensible heuristics. Finally, it is revealed that the optimal end positions of communication phases become increasingly independent of the data payload, for large data payload values.
\end{abstract}
\vspace{-1mm}
\begin{IEEEkeywords}
UAV-assisted wireless networks, adaptive trajectory optimization, semi-Markov decision process
\end{IEEEkeywords}
\vspace{-4mm}
%%%%%%%%%% Introduction %%%%%%%%%%
\section{Introduction} \label{section_intro}
%What is the fundamental problem?
Recently, much research has gone into UAVs operating in wireless networks \cite{FundTrade,EnEffDeployment,EnMin,TrajOptRL,OptTransport}. The drive for this is due to the unique benefits that UAVs acting as flying base stations, mobile relays, etc., provide in enhancing the overall network performance, thanks to their unique advantages over terrestrial counterparts in terms of mobility, maneuverability, and improved line-of-sight (LoS) link probability \cite{FundTrade}. However, the design of UAV deployment strategies comes with challenges, namely the determination of optimal positioning or trajectories in the face of constraints imposed on UAV energy consumption, network throughput, and/or delay requirements \cite{FundTrade,EnEffDeployment,EnMin,TrajOptRL}.

%this paragraph briefly summarizes research efforts
Some research has focused on the trajectory optimization under energy constraints, as in \cite{EnEffDeployment} and \cite{EnMin}. In \cite{LOSMapApproach}, the fine-grained structure of LoS conditions is exploited to position UAVs optimally with the goal to maximize throughput. In \cite{TrajOptRL}, a model-free Q-learning approach is taken in the trajectory design so as to maximize the transmission sum-rate.
%%%%

All of these efforts consider situations that are solved in the \emph{offline} case, i.e.,  the pattern of transmission requests is known in advance, so that the trajectory may be pre-planned accordingly. However, this may be impractical as transmission requests are often random and cannot be determined in advance. In these cases, trajectory design is much more challenging, since it must be continuously adjusted based on the realization of these random processes, and incorporate the uncertainty in the future evolution of the system dynamics. In this paper, we investigate this problem by developing policies that adapt the trajectory based on the random realization of downlink transmission requests generated by two ground nodes (GNs), so as to optimize the average long-term performance.

To further motivate the need for this new formulation, consider the scenario depicted in Fig. \ref{fig:SystemModel}.
In this context, the minimum communication delay to serve GN$_1$ is achieved by flying towards it to improve the distance-dependent pathloss. With this design, for a sufficiently large data payload, the UAV will terminate the data transmission hovering above GN$_1$, where channel conditions are most favorable. However, if the UAV is to service a random request generated by GN$_2$ shortly after terminating the transmission to GN$_1$, the delay incurred to serve this second request may be large due to the large distance that separates the UAV from GN$_2$, causing severe pathloss conditions. In other words, 
 under random transmission requests, the greedy delay minimization to serve a certain request may lead the UAV to a position where subsequent random requests cannot be served effectively, yielding poor delay performance in an average long-term sense. This example points to the need to incorporate the random nature of transmission requests in the trajectory design.
  
To address this question, we consider a scenario in which an UAV is serving two GNs far apart, and receives transmission requests according to a Poisson random process. We formulate the problem as that of designing an \emph{adaptive} trajectory, so as to minimize the average long-term communication delay incurred to serve the requests of both GNs. We prove that the optimal trajectory in the communication phase operates according to a two-scale optimization: in the \emph{outer optimization}, the UAV selects a target end position, 
which optimizes the trade-off between minimizing the delay of the current request, and minimizing the expected average long-term delay; then, in the \emph{inner optimization}, the UAV travels greedily from the current position to the selected end position while communicating, following
the trajectory that greedily minimizes the communication delay for the current request, provided in closed form.
We utilize a multi-chain policy iteration algorithm to
optimize the selection of the end position in the \emph{communication phase} and the 
trajectory during the \emph{waiting phase}, in which the UAV is
not actively servicing downlink transmission requests. Our numerical results reveal that the UAV should always move towards the geometric center of the two GNs during the waiting phase, and that the optimal trajectory during communication phases becomes increasingly independent of the data payload and only determined by system parameters as the data payload value becomes sufficiently large.

The rest of the paper is organized as follows. In Sec. \ref{section_sys}, we introduce the system model and state the optimization problem; in Sec. \ref{analysis_pol}, we formalize the problem as a semi-Markov decision process (SMDP); in Sec. \ref{numerical_results}, we provide numerical results; lastly, in Sec. \ref{discussion}, we conclude the paper with some final remarks.
\vspace{-2mm}
\section{System Model and Problem Formulation} \label{section_sys}
\subsection{System Model}
Consider the scenario depicted in Fig. \ref{fig:SystemModel}, where one rotary-wing UAV services two ground nodes (GNs) with random downlink\footnote{This formulation and the analysis can be directly applied to uplink transmissions as well.} transmission requests of $L$ bits. The two ground units GN$_1$ and GN$_2$ are located at positions $x_1{=}-a$ and $x_2{=}a$ along the $x$-axis, respectively, both at ground level (height $0$). The UAV moves along the line segment connecting the two GNs, at height $H$ from the ground. We let $q(t){\in}[-a,a]$ be the UAV's position along the $x$-axis at time $t$, and we assume that it is either hovering or moving at speed $V$, hence $|q^\prime(t)|{\in}\{0,V\}$, where $f^\prime$ denotes derivative of $f$ over time.

\begin{comment} %leaving out for the time being and I'll add in a PowerPoint image
\begin{figure}[!htbp]
\begin{centering}
\begin{tikzpicture}
\draw[step=1cm, gray, very thin] (-4,-1) grid (4,2);

\draw[thick, gray, ->] (-4,0) -- (4.2,0) 
node[anchor=north west]{\small $x$};

\draw[thick, gray, ->] (0,-1) -- (0,2.2)
node[anchor = south east]{\small $z$};

\filldraw[blue] (-3.1,-.1) rectangle (-2.9,.1)
node[below=.5cm, left=.1cm, black]{\small $w_1$};

\filldraw[blue] (3.1,-.1) rectangle (2.9,.1)
node[below=.5cm, right=.1cm, black]{\small $w_2$};

\filldraw[red] (1,1.2) circle (3pt)
node[above = .5cm, right = .1cm, black]{\small UAV};

\draw[red, thick, dashed, ->] (1,1.2) -- (-2.5,.2);

\draw[black, thick] (1.4,0) -- (1.4,1.2);

\draw[black, thick] (1.35,1.2) -- (1.45,1.2)
node[below = .5cm, right = .03cm, black]{\small $H$};

\draw[green, thick, ->] (1,1.2) -- (-.5,1.2)
node[above = .4cm, left = .1cm, black]{\small $q(t)$};

\end{tikzpicture}
\caption{$x$-$z$ View of UAV and Two GNs.
\nm{The picture looks good, but I would make it more "user-friendly": replace the red dot with an actual UAV, and the two GNs with some sort of transmitter. It is easier if you create the picture with powerpoint instead of tikzpicture.}
}
\label{fig:GNVisual}
\end{centering}
\end{figure}
\end{comment}

A base station (BS) connected to the rest of the network is the source of downlink traffic to the two GNs. When a downlink request is generated by a certain GN, the BS transmits the data payload to the UAV, which then relays it to the GN using a decode and forward strategy \cite{1056084}. We assume that the UAV has a high-capacity link to the BS, hence the communication link between the UAV and the GN constitutes the bottleneck of the overall BS-UAV-GN communication. In the rest of the paper, we thus focus on the UAV-GN communication and neglect the delay over the BS-UAV link.
We assume that the UAV transmits at fixed power $P_c$ and that the communication intervals experience LoS links with no probabilistic elements. This is motivated by the fact that UAVs in low-altitude platforms  generally tend to have a much higher occurrence of LoS links \cite{LOSDominance}.
We model the instantaneous communication rate between the UAV in position $q(t)$ and GN$_r, r\in\{1,2\}$ in position $x_r$ as
\begin{equation}\label{eq:RateEqn}
	 R_{r}(q(t)) = B\log_{2} \Bigg(1+ \frac{\gamma}{H^2 + (q(t)-x_r)^2} \Bigg),
\end{equation}
where $H^2 + (q(t)-x_r)^2$ is the squared distance between the UAV and GN$_r$, $B$ is the channel bandwidth, and $\gamma$ is the SNR referenced at $1$ meter (see \cite{EnMin}).

When the UAV has no active transmission requests, future requests arrive according to a Poisson process with mean $\lambda/2$ requests/second, independently at each GN. Each request requires the transmission of $L$ bits to the corresponding destination. Upon receiving a request from GN$_r$, the UAV enters the \emph{communication phase}, where it services it by transmitting the $L$ bits to GN$_r$; any additional requests received during this communication interval are dropped (see also Fig. \ref{fig:SystemModel}). After the data transmission is completed, the UAV enters the \emph{waiting phase}, where it awaits for new requests (with rate $\lambda/2$ for each GN), and the process is repeated indefinitely. During this periodic process of communication and waiting for new requests, the UAV follows a trajectory, part of our design, with the goal to minimize the average long-term communication delay, as discussed next.
   
\begin{figure}[t]
\centering
\includegraphics[width=8.0cm,trim={1mm 1mm 1mm 1mm},clip]{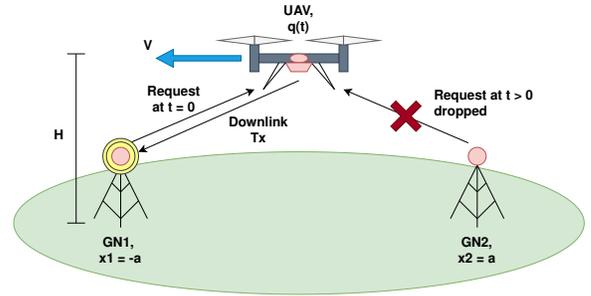}
\caption{System model depicting a downlink transmission request from GN$_1$; the request from GN$_2$ is dropped during the active communication interval.}
\label{fig:SystemModel}
\vspace{-5mm}
\end{figure}

%I think we said this in the intro clearly enough:::Different from previous work which studies \emph{offline} trajectory optimization, (see for instance \cite{FundTrade,EnEffDeployment,EnMin,TrajOptRL}), in this work the trajectory followed by the UAV is the outcome of an \emph{online decision making process}, as a result of the random request process and service duration of the UAV.
    
\subsection{Problem Formulation}
In this work, we consider the unconstrained delay minimization and neglect the propulsion energy consumption from our problem. In fact, it has been shown that a rotary-wing UAV exhibits comparable energy consumption when either moving or hovering \cite{EnMin}; in the special case when the moving and hovering powers are equal (for instance, based on the model in \cite{EnMin}, this occurs at speed $V{=} 38$ m/s), the finite energy in the UAV battery translates into a constraint on the total service time of the UAV, independent of the trajectory followed. %no need to discuss too much of future work? also to total space..The extension to more general propulsion energy models will be part of our future work.

The goal is to define the optimal policy (UAV trajectory) so as to minimize the average communication delay. To this end, let $\Delta_l$ be the delay incurred to complete the transmission of the $l$th request serviced by the UAV. Let $M_t$ be the total number of requests served and completed up to time $t$. Then, we define the expected average delay under a given trajectory policy $\mu$ (to be defined), starting from $q(0)=0$ as\footnote{While in practice the operation time of the UAV is constrained by the amount of energy stored in its battery, and the policy \emph{should} depend on the amount of time left, the asymptotic case $t\to\infty$ is convenient since it gives rise to \emph{stationary policies} (i.e., time-independent); this is a good approximation when the dynamics of the waiting and communication phases occur at much faster time scales than the total travel time, i.e., when $M_t$ in \eqref{eq:OverallObj} is large for practical values of the travel time $t$. For perspective, \cite{HoverTime} places typical rotary-wing hovering endurance times in the 15-30 minute range.}
\begin{equation}\label{eq:OverallObj}
	\bar{D}_{\mu} = \lim_{t \rightarrow \infty} \mathbb{E} \Bigg[\frac{\sum_{l = 0}^{M_t - 1} \Delta_l}{M_t}\Bigg|q(0)=0\Bigg].
\end{equation}
We then seek to determine $\mu^*$ to minimize $\bar{D}_{\mu}$, i.e.,
\begin{equation}\label{eq:OverallPol}
	\mu^* = \arg\min_\mu\; \bar{D}_{\mu}.
\end{equation}

Note that this is a non-trivial optimization problem. While the minimum delay to serve a request, say from GN$_1$, is achieved by flying towards GN$_1$ at maximum speed to improve the link quality, this strategy may not be optimal in an \emph{average} delay sense: if the UAV receives a new request from GN$_2$ shortly after completing the request to GN$_1$, the delay to serve this second request may be large due to the large distance that must be covered by the UAV.
%\emph{Thus, there is a trade-off between minimizing the delay of the current request and minimizing that of future requests.}
\vspace{-3mm}
\subsection{Semi-Markov Decision Process (SMDP) formulation}
In general, a solution to (\ref{eq:OverallPol}) would involve the optimization of an intractable number of variables over time (i.e., all possible trajectories followed by the UAV at any given time), over a continuous state space (the interval $[-a,a]$). Therefore, it is advantageous to approximate the system model through discretization and reformulate (\ref{eq:OverallPol}) as an average-cost SMDP.

We define the state space as $\mathcal{S}{=}\mathcal{I}{\times}\mathcal{R}$, where
$\mathcal{R}{=}\{0,1,2 \}$ denotes the request status, i.e., no active request ($0$), or a request is received from GN$_r$ ($r\in\{1,2\}$), and
\begin{equation}\label{eq:QStateDefinition}
\mathcal{I} \triangleq \left\{-N,N+1,\dots,N-1,N\right\}
\end{equation}
is the set of $2N+1$ indices corresponding to discretized positions $\mathcal Q\triangleq\{q_i=\frac{i}{N}a,\ \forall i\in\mathcal I\}$ along the interval $q(t) \in [-a,a]$. This is a good approximation for sufficiently large $N$, as
$\frac{a}{NV}\lambda\ll 1$, making the expected number of requests received over the travel time between two adjacent discretized positions much smaller than one. It is also useful to further partition the state space into \emph{waiting states},
 $\mathcal S_{\mathrm{wait}}=\mathcal I\times\{0\}$, and \emph{communication states}, $\mathcal S_{\mathrm{comm}}=\mathcal I\times\{1,2\}$.

To define this SMDP, we sample the continuous time interval to define a discrete sequence of states $\{s_n,n\geq 0\}\subseteq\mathcal S$ with the Markov property. We now define the actions available, the transition probabilities, duration and cost of each state visit.

If the UAV is in state $s_n{=}(i,0){\in}\mathcal S_{\mathrm{wait}}$ at time $t$, i.e., it is in the discretized position $q_i$ and there are no active requests, then the actions available are $a_n{\in}\{-1,0,1\}$, i.e.
move right ($a_n{=}1$ to position $q_{i{+}1}$), hover ($a_n{=}0$), or move left by one discretized position ($a_n{=}{-}1$ to $q_{i{-}1}$). The amount of time required to take this action, i.e.,
to fly between two adjacent discretized positions, is
\begin{equation}\label{eq:DeltaDefinition}
	\Delta_{0}\triangleq\frac{a}{NV}.
\end{equation}
The new state is then sampled at time $t{+}\Delta_0$, and is given by
$s_{n+1}{=}(i{+}a_n,r_{n+1})$. The transition probability from state $s_n{=}(i,0)$ under action $a_n\in\{-1,0,1\}$ is then given by
\begin{align}\label{eq:Req0TransitionProb}
&\mathbb P(s_{n+1}{=}(i{+}m,0)|s_n{=}(i,0),a_n{=}m)=e^{-\lambda\Delta_{0}},
\\
&\mathbb P(s_{n+1}{=}(i{+}m,r)|s_n{=}(i,0),a_n{=}m){=}\frac{1{-}e^{-\lambda\Delta_{0}}}{2},\forall\, r{\in}\{1,2\},
\nonumber
\end{align}
depending on whether no request is received during this time interval ($r_{n+1}{=}0$, with probability $e^{-\lambda\Delta_{0}}$), or a request is received from GN$_r$ ($r_{n+1}{=}r\in\{1,2\}$, with probability $[1-e^{-\lambda\Delta_{0}}]/2$ for each GN).

Upon reaching state $s_n{=}(i,r){\in}\mathcal S_{\mathrm{comm}}$ with $r{\in}\{1,2\}$ at time $t$, the UAV has received a request to serve $L$ bits to GN$_r$. The actions available at this point are all trajectories that start from $q_i$ and allow the UAV to transmit the entire data payload of $L$ bits. Assuming a \emph{move and transmit} strategy (see \cite{EnMin}), the selected trajectory $q(\cdot)$ of duration $\Delta$ must satisfy
\begin{equation}\label{eq:RateRequirement}
\int_{0}^{\Delta}R_{r}(q(\tau))d\tau \geq L,
\end{equation}
since all bits need to be transmitted during this phase. Under this trajectory, the communication delay is thus $\Delta$. We define the action space in state $(i,r)\in \mathcal{S}_{\mathrm{comm}}$ as the set of all feasible trajectories,
$\mathcal T_r (i)=\cup_j\mathcal T_r(i\to j)$,
where we have defined $\mathcal T_r(i\to j)$ as the set of feasible trajectories starting in 
${q}_i$, ending in ${q}_j$, and serving GN$_r$, i.e.,
\begin{equation}\label{eq:CommActionTrajSet}
\mathcal T_r(i\to j)
=\biggr\{q:[0,\Delta]\to[-a,a]:\int_{0}^{\Delta}R_{r}(q(\tau))d \tau \geq L,$$$$
|q^\prime(t)|\leq V,\ 
q(0)={q}_i,\ q(\Delta)={q}_j,\ \exists\Delta>0\biggr\}.
\end{equation}
Upon completing the communication phase, the UAV enters the waiting phase again; the new state is then sampled at time $t{+}\Delta$ (the amount of time elapsed to complete the selected trajectory), and is given by
$s_{n+1}{=}(j,0){\in}\mathcal S_{\mathrm{wait}}$, corresponding to
the position ${q}_j$ reached at the end of the communication phase. Thus, we have defined the transition probability in the SMDP from state $s_n=(i,r)$ under action $q\in\mathcal T_r(i\to j)$ as
\begin{align}
\label{eq:CommStateTransition}
&\mathbb P(s_{n+1}=(j,0)|s_n=(i,r),q)=1,\ \forall q\in\mathcal T_r(i\to j).
\end{align}
In other words, the trajectory selection process in the communication phase can be
described via a two-scale decision process: 1) given $(i,r)$, i.e., the current position $q_i$ of the UAV and the request received from GN$_r$, the UAV first selects some $j\in \mathcal I$, which defines the target position $q_j$ to be reached at the end of the communication phase; 2) the UAV selects a feasible trajectory $q$ from $\mathcal T_r(i\to j)$, executes the trajectory while communicating to GN$_r$, and terminates the communication phase in the new position $q_j$, corresponding to state $(j,0)$. After this point, the UAV is in the waiting phase again.
%\nm{The following is not very clear to me, I rewrote it as above.}

With the states and actions defined, we can define a policy $\mu$. Specifically, for states $(i,0) \in \mathcal{S}_{\mathrm{wait}}$, $\mu(i,0)\in \{-1,0,1\}$. Likewise, for states $(i,r) \in \mathcal{S}_{\mathrm{comm}}$, $\mu(i,r) = (j,q(\cdot))$, where $j\in \mathcal I$ (position reached at the end of the communication phase) and $q(\cdot)\in\mathcal T_r(i\to j)$ (feasible trajectory starting in $q_i$, ending in $q_j$, to serve GN$_r$).

The \emph{communication delay} cost during the \emph{waiting phase} is zero, i.e. $\Delta_{i,0}(m)=0$, for all states $(i,0)\in\mathcal S_{\mathrm{wait}}$ and actions $m\in\{-1,0,1\}$. When the UAV is in a \emph{communicating phase}, we denote the communication delay incurred in state $(i,r)$ under action $(j,q(\cdot))$ as $\Delta_{i,r}(j,q(\cdot))$. Compactly, we write $\Delta_s(\mu(s))$ to denote the delay incurred in state $s \in \mathcal S$ under the action $\mu(s)$ dictated by policy $\mu$.

With this notation, and having now defined a \emph{stationary policy} $\mu$, we can rewrite the average delay $\bar{D}_{\mu}$ in \eqref{eq:OverallObj} in the context of the SMDP as
\begin{equation}\label{eq:OverallObjH}
\bar{D}_{\mu}{=}\lim_{K \rightarrow \infty} \mathbb{E} \Bigg[\frac{
\frac{1}{K}\sum_{n= 0}^{K - 1} \Delta_{s_n}(\mu(s_n))}{
\frac{1}{K}\sum_{n= 0}^{K - 1}\chi(s_n\in\mathcal S_{\mathrm{comm}})
}\Bigg|s_0=(0,0)\Bigg],
\end{equation}
where $\chi(A)$ is the indicator function of the event $A$.
In fact, the numerator in \eqref{eq:OverallObj} counts the sample average delay incurred in the communication phases up to slot $K$ of the SMDP, whereas the denominator in \eqref{eq:OverallObj} counts the sample average number of communication slots in the SMDP up to slot $K$.
Now, using Little's Theorem \cite{LittlesTheorem}, 
we can rewrite \eqref{eq:OverallObjH} as
\begin{align}\label{eq:ExpAvgDelay}
\bar{D}_{\mu}&{=}\frac{\sum_{s \in \mathcal{S}}\Pi_{\mu}(s) \Delta_{s}(\mu(s))}{
\sum_{s \in \mathcal{S}}\!\!\Pi_{\mu}(s)\chi(s{\in}\mathcal{S}_{\mathrm{comm}})}
%\frac{\underset{s \in \mathcal{S}}{\sum}\Pi_{\mu}(s) \Delta_{s}(\mu(s))}{\underset{s \in \mathcal{S}}\sum\Pi_{\mu}(s)\chi(s{\in}\mathcal{S}_{\mathrm{comm}})}
\nonumber\\&
{=}\frac{\sum_{s \in \mathcal{S}_{\mathrm{comm}}}\Pi_{\mu}(s) \Delta_{s}(\mu(s))}{\sum_{{s \in \mathcal{S}_{\mathrm{comm}}}}\Pi_{\mu}(s)},
\end{align}
where $\Pi_{\mu}(s)$ is the steady-state probability in the SMDP of the UAV being in state $s$ under policy $\mu$, and the second equality holds since $\Delta_{s}(\mu(s))=0$ and  $\chi(s{\in}\mathcal{S}_{\mathrm{comm}}) = 0$ for $s\in\mathcal{S}_{\mathrm{wait}}$.

\section{Policy Optimization and Analysis}\label{analysis_pol}
In this section, we tackle the solution to the optimization problem \eqref{eq:OverallPol}, with $\bar{D}_{\mu}$ given by \eqref{eq:ExpAvgDelay}.
However, \eqref{eq:OverallPol} cannot be directly solved using dynamic programming techniques, due to the presence of the denominator in \eqref{eq:ExpAvgDelay}, which depends on the policy selected $\mu$, hence it affects the optimization. The next lemma demonstrates that the denominator of (\ref{eq:ExpAvgDelay}) can be expressed as a positive constant, \emph{independent} from policy $\mu$ and only dependent on system parameters. In doing so, the optimization of $\mu$ only needs to focus on the minimization of $\sum_{s \in \mathcal{S}}\Pi_{\mu}(s) \Delta_{s}(\mu(s))$, so that \eqref{eq:OverallPol} can be cast as an \emph{average cost per stage problem}, solvable with standard dynamic programming techniques.
\begin{lemma}\label{eq:SSComm}
	Let $\pi_{\mathrm{wait}}$ and $\pi_{\mathrm{comm}}$ be the steady-state probabilities that the UAV is in the waiting and communication phases,
$\pi_{\mathrm{comm}}{=}\sum_{s \in \mathcal{S}_{\mathrm{comm}}}\!\!\!\!\!\Pi_{\mu}(s)$
and $\pi_{\mathrm{wait}}{=}1{-}\pi_{\mathrm{comm}}$. We have that
\begin{equation}\label{eq:SSCommResult}
	\pi_{\mathrm{wait}} = \frac{1}{2-e^{-\lambda \Delta_0}}, \; \pi_{\mathrm{comm}} = \frac{1-e^{-\lambda \Delta_0}}{2-e^{-\lambda \Delta_0}}.
\end{equation}
\end{lemma}
\begin{proof}
Let $p_{ww}$, $p_{wc}$, $p_{cw}$, and $p_{cc}$ be the probabilities of a state request status, $r \in \mathcal{R} = \{0,1,2\}$, transitioning in the SMDP as $0{\rightarrow}0$, $0{\rightarrow}\{1,2\}$, $\{1,2\}{\rightarrow}0$, and $\{1,2\}{\rightarrow}\{1,2\}$, respectively. Then, $p_{ww} = e^{-\lambda \Delta_0}$ (if no request is received, the SMDP remains in the waiting state), $p_{wc}=1-p_{ww}$, $p_{cw} = 1$, and $p_{cc} = 0$ (if the SMDP is in the communication state, the next state of the SMDP will be a waiting state, see \eqref{eq:CommStateTransition}). Therefore, the steady-state probabilities of  being in the waiting and communication states, $\pi_{\mathrm{wait}}$ and $\pi_{\mathrm{comm}}$, satisfy
\begin{align*}
&		\pi_{\mathrm{wait}} = p_{ww}\pi_{\mathrm{wait}} + p_{cw}\pi_{\mathrm{comm}} \nonumber
		=e^{-\lambda \Delta_0}\pi_{\mathrm{wait}} + \pi_{\mathrm{comm}},
		\\&
		\pi_{\mathrm{comm}} = p_{wc}\pi_{\mathrm{wait}} + p_{cc}\pi_{\mathrm{comm}}
		=(1-e^{-\lambda \Delta_0})\pi_{\mathrm{wait}}, \nonumber
\\&
		\pi_{\mathrm{wait}} + \pi_{\mathrm{comm}} = 1,
\end{align*}
whose solution is given as in the statement of the lemma.
\end{proof}

When we refer to the denominator of (\ref{eq:ExpAvgDelay}), it is evident that it is equal to the steady-state probability that the UAV is in a communication state while following policy $\mu$, $\pi_{\mathrm{comm}}$. However, with the result of Lemma \ref{eq:SSComm}, $\pi_{\mathrm{comm}}$ is simply a positive constant determined by system parameters, yielding
\begin{equation}\label{eq:ExpAvgDelaySimplified}
	\bar{D}_{\mu} = \frac{\sum_{s \in \mathcal{S}}\Pi_{\mu}(s) \Delta_{s}(\mu(s))}{\pi_{\mathrm{comm}}},
\end{equation}
which we now aim to minimize with respect to policy $\mu$.

As the problem stands now, the \emph{communication} phase selects an action from $\mathcal{T}_r(i)$, which is a set containing an uncountable number of trajectories. By exploiting the two-scale structure of the problem outlined earlier, we now demonstrate that only a finite set of trajectories from $\mathcal{T}_r(i)$ are eligible to be optimal, for each state $(i,r) \in \mathcal{S}_{\mathrm{comm}}$, hence making the problem a \emph{finite state and action} SMDP.
\vspace{-3mm}
\subsection{Decomposition of Policy $\mu$}
Note from \eqref{eq:CommStateTransition} that the transition probability from a communication state $s_n{=}(i,r)$ under action $(j,q(\cdot))$ is only affected by the selection of $j$ and not the particular trajectory $q(\cdot){\in}\mathcal{T}_r(i \rightarrow j)$ that leads from $q_i$ to $q_j$ during the communication phase. It follows that the steady-state probability $\Pi_{\mu}$ under $\mu(i,r){=}(j,q(\cdot))$ is only affected by the selection of $j$ and not the specific trajectory within $\mathcal{T}_r(i \rightarrow j)$.

By establishing this property, we decompose the policy $\mu$ into the \emph{waiting policy} $\theta(i){\in}\{-1,0,1\}$, which defines the optimal action in state $(i,0){\in}\mathcal S_{\mathrm{wait}}$ of the waiting phase;
the \emph{end position policy} $J(i,r)$, which selects the end position $q_j$ with $j{=}J(i,r)$ to be reached at the end of the communication phase; and the \emph{trajectory policy} $\rho(i,r,j)$ which, given $j{=}J(i,r)$, selects a trajectory $q(\cdot){=}\rho(i,r,j)$ from $\mathcal T_r(i\to j)$. 
Owing to the independence of $\Pi_{\mu}$ on the trajectory policy $\rho$, the delay minimization problem can then be rewritten as
\begin{equation*}
	\bar{D}_{\mu}^* = \frac{\min\limits_{\theta,J}\, \sum\limits_{s \in \mathcal{S}_{\mathrm{comm}}}  \, \Pi_{\theta,J}(s)\min\limits_{\rho(s,J(s)))} \Delta_{s}(J(s),\rho(s,J(s)))}{\pi_{\mathrm{comm}}}.
\end{equation*}
Letting
\begin{align}\label{eq:OfflineTrajOpt}
\Delta_r^*(i,j){\triangleq}\!\!\!\!\!\!\min\limits_{q(\cdot) \in \mathcal{T}_r(i \rightarrow j)}\!\!\!\!\!\Delta_{i,r}(j,q), \;\; \forall\ (i,r){\in}\mathcal S_{\mathrm{comm}},\forall j\in\mathcal I,
\end{align}
%be the minimum communication delay incurred to travel from $q_i$ to $q_j$ while serving GN$_r$,
we can finally write
\begin{align}\label{eq:FinalMinSimplification2}
	\bar{D}_{\mu}^* = \frac{\min\limits_{\theta,J}\, \sum_{(i,r) \in \mathcal{S}_{\mathrm{comm}}}  \, \Pi_{\theta,J}(i,r)\Delta_r^*(i,J(i,r))}{\pi_{\mathrm{comm}}}.
\end{align}

Note that $\Delta_r^*(i,j)$ yields the trajectory
that greedily minimizes the communication delay when 
 starting from state $(i,r)$, ending in position $q_{j}$ while serving GN$_r$.
 This result proves that, for any communication state $(i,r)$, there exist only 
 $2N+1$ trajectories that are eligible to be optimal, one for each possible ending position $q_j\in\mathcal Q$. Hence, the problem is finally reduced to that of finding the optimal waiting policy $\theta$ and end position policy $J$, which can be solved efficiently via dynamic programming (Algorithm \ref{alg:PI}). In the next section, we provide a closed form expression of $\Delta_r^*(i,j)$.
\vspace{-3mm}
\subsection{Closed-form Delay Minimizing Trajectory}
With the  independence of the steady-state probabilities from $\rho$, we can proceed to solve \eqref{eq:OfflineTrajOpt} and then provide the dynamic programming algorithm to solve for $\theta^*$ and $J^*$ in \eqref{eq:FinalMinSimplification2}. By definition of $\mathcal{T}_r(i {\rightarrow} j)$ in \eqref{eq:CommActionTrajSet}, 
we can rewrite $\Delta_r^*(i,j)$ as
\begin{align}\label{eq:TrajMinDelta}
\nonumber
	\Delta_r^*(i,j)= \min\limits_{\Delta,q} &\Big\{\Delta \,\Bigr| \, \int_{0}^{\Delta} R_{r}(q(\tau))d\tau \geq L, \\& |q^\prime(\tau)| \leq V, q(0) = q_i, q(\Delta) = q_j \Big\}.
\end{align}
The minimizer $q^*$ is the trajectory that the UAV should follow when receiving a request from GN$_r$ starting in position $q_i$ and 
ending in position $q_j$, selected by the end position policy $J$.

In defining the optimal trajectory, the following definitions will be useful. Let 
$\tau_{p_1,p_2}{\triangleq}\frac{|p_2-p_1|}{V}$
be the time needed to fly at maximum speed from $p_1$ to $p_2{\in}[-a,a]$. Along this straight trajectory, let
\begin{equation}\label{eq:ellFirst}
\ell_{p_1,p_2}^{(r)} \triangleq \int_0^{\tau_{p_1,p_2}}R_{r}\left(p_1+\frac{\tau}{\tau_{p_1,p_2}}(p_2-p_1)\right)d\tau
\end{equation}
be the amount of bits transmitted to serve GN$_r$.
% when moving at maximum speed from $p_1$ to $p_2$, when serving GN$_r$.

 Clearly, $\ell_{p_1,p_1}^{(r)}{=}0$
  ($\tau_{p_1,p_1}{=}0$), $\ell_{p_1,p_2}^{(r)}{=}\ell_{p_2,p_1}^{(r)}$ ($\tau_{p_1,p_2}{=}\tau_{p_2,p_1}$), and $\ell_{p_1,p_2}^{(1)}{=}\ell_{-p_1,-p_2}^{(2)}$ ($\tau_{p_1,p_2}{=}\tau_{-p_1,-p_2}$).
The integral $\ell_{p_1,p_2}^{(r)}$ can be determined in closed form and is found in \cite{EnEffDeployment}, for example.
%(see also \cite{EnEffDeployment})
%\begin{align}\label{eq:ClosedFormRateIntegralQjgQi}
%	&\ell_{p_1,p_2}^{(r)}=\frac{-\alpha B \tau_{p_1,p_2}}{(p_1-p_2)\ln(2)} \Bigg\{\ln \Big(1+\frac{\gamma}{H^2 + \alpha^2} \Big)\\
%	&- 2 \tau_{p_1,p_2} \sqrt{\gamma + H^2} \tan^{-1}{\Big(\frac{\alpha}{\sqrt{\gamma+H^2}} \Big)} \\
%	\nonumber
%	 & -2H \tau_{p_1,p_2} \tan^{-1}{\Big(\frac{\alpha}{H} \Big)} \Bigg\} \Bigg|_{\alpha = p_1-w_r}^{\alpha = p_2-w_r}.
%	\nonumber
%\end{align}
%where $f_\alpha|_{x}^y=f(y)-f(x)$. 
We also define the trajectory $\upsilon\{p_1{\to}(p_2,\delta){\to}p_3\}(\tau),\ \tau\in[0,\delta{+}\tau_{p_1,p_2}{+}\tau_{p_2,p_3}]$,
 as the one in which the UAV starts at position $p_1$, flies at maximum speed to
$p_2$, hovers at $p_2$ for $\delta$ amount of time, and finally flies at maximum speed from $p_2$ to $p_3$. Mathematically,
	\begin{align}
&	\upsilon\{p_1\to (p_2,\delta)\to p_3\}(\tau)\\&{=}\!
	\begin{cases}
	p_1+\frac{\tau}{\tau_{p_1,p_2}}(p_2-p_1),
	\!\!&\tau{\in}[0,\tau_{p_1,p_2}]\\
	p_2,\!\!&\tau{\in}[\tau_{p_1,p_2},\tau_{p_1,p_2}+\delta]\\
	p_2{+}\frac{\tau{-}\tau_{p_1,p_2}{-}\delta}{\tau_{p_2,p_3}}(p_3{-}p_2),
	\!\!&\tau{\in}[\tau_{p_1,p_2}{+}\delta,\tau_{p_1,p_2}{+}\tau_{p_2,p_3}{+}\delta].
	\end{cases}
	\nonumber
	\end{align}
The traffic delivered to GN$_r$ when following this trajectory is
$\ell_{p_1,p_2}^{(r)}\!\!{+}\delta R_{r}(p_2){+}\ell_{p_2,p_3}^{(r)}$, with delay $\tau_{p_1,p_2}\!\!{+}\delta{+}\tau_{p_2,p_3}$.
With these definitions, we are now ready to state the main result.
\begin{theorem}\label{th:OptTrajCase1}
	Let $q^*(\cdot){\in}\mathcal{T}_r(i \rightarrow j)$ be the trajectory that minimizes the communication delay $\Delta_r^*(i,j)$. 
	If $\ell_{q_i,q_j}^{(r)}{\geq}L$, then
\begin{equation}\label{eq:QOptForGEQL}
	q^*(\cdot)=\upsilon\{q_i\to (q_j,0)\to q_j\}(\cdot),\ \Delta_r^*(i,j)=\tau_{q_i,q_j},
\end{equation}
	i.e., the UAV flies at maximum speed from $q_i$ to $q_j$ without interruption; otherwise, if $\ell_{q_i,x_r}^{(r)}+\ell_{x_r,q_j}^{(r)} \leq L$, then
\begin{equation*}
q^*(\cdot){=}\upsilon\{q_i\to (x_r,\delta^*){\to}q_j\}(\cdot),\ \Delta_r^*(i,j){=}\tau_{q_i,x_r}{+}\tau_{x_r,q_j}{+}\delta^*,
\end{equation*}
where 
\begin{equation}\label{eq:DStarDef}
\delta^*=\frac{L-\ell_{q_i,x_r}^{(r)}-\ell_{x_r,q_j}^{(r)}}{R_{r}(x_r)};
\end{equation}
i.e., the UAV flies at maximum speed from $q_i$ to $x_r$, hovers over $x_r$ for $\delta^*$ amount of time, and then flies to $q_j$;
finally, if $\ell_{q_i,x_r}^{(r)}+\ell_{x_r,q_j}^{(r)} > L$, but $\ell_{q_i,q_j}^{(r)} < L$, then
\begin{equation*}
q^*(\cdot)=\upsilon\{q_i\to (p^*,0)\to q_j\}(\cdot),\ \Delta_r^*(i,j)=\tau_{q_i,p^*}+\tau_{p^*,q_j},
\end{equation*}
where $p^*$ is the unique solution in 
$[x_r,\min\{q_i,q_j\}]$ (if $r{=}1$)
or
$[\max\{q_i,q_j\},x_r]$ (if $r{=}2$)
 of
$\ell_{q_i,p^*}^{(r)}{+}\ell_{p^*,q_j}^{(r)}{=}L$;
i.e., the UAV flies at maximum speed towards $x_r$ to the farthest point $p^*$ and then back to $q_j$, with $p^*$ uniquely defined in such a way as to transmit exactly the data payload upon reaching $q_j$.
\end{theorem}
\begin{proof}\label{pf:minTrajectory}
Due to space limitations, we provide an outline of the proof. 
Assume $r{=}2$ (a similar argument applies to $r{=}1$ by symmetry).
1) for any trajectory $q(\cdot)\in\mathcal T_2(i\to j)$ of duration $\Delta$, one can construct another trajectory $\tilde q(\cdot)\in\mathcal T_2(i\to j)$
of same duration $\Delta$, and such that $|q(t) - x_r| \geq |\tilde q(t) - x_r|,\ \forall t\in[0,\Delta]$; such trajectory is obtained by flying at maximum speed towards GN$_2$, possibly hovering on top of GN$_2$ for $\delta$ amount of time (if time allows), and then returning to $q_j$, yielding $\tilde q(\cdot){=}\upsilon\{q_i{\to}(p^*,\delta^*){\to}q_j\}(\cdot)$, for a proper choice of $p^*$ and $\delta^*$ such that $\tau_{q_i,p^*}{+}\tau_{p^*,q_j}{+}\delta{=}\Delta$;
2) note that the UAV is always closer to GN$_2$ under $\tilde q(\cdot)$ than it is under $q(\cdot)$, hence it delivers a larger data payload than $q(\cdot)$ while incurring the same delay; therefore, $q(\cdot)$ is suboptimal;
3) $\tilde q(\cdot)$ can be further improved by minimizing the delay (by optimizing $(p^*,\delta^*)$), yielding the three cases provided in the statement of the theorem.% Full proof!!!
\end{proof}
\vspace{-3mm}
\subsection{Multi-chain Policy Iteration (PI) Algorithm}
We opt to use a multi-chain PI algorithm to solve \eqref{eq:FinalMinSimplification2}, as
 there exist some policies whose induced Markov chain structures are multi-chain. For example, if the \emph{waiting policy} is $\theta(i) = 0$, and the \emph{end position policy} is $J(i,r) = i$, then the induced Markov chain has $2N+1$ recurrent classes (hence multi-chain). To accommodate this structure, the pseudocode that follows is based upon the multi-chain PI methods of \cite{PutermanMDP} and succinctly describes how to solve for $\mu^*$.

In Algorithm \ref{alg:MultichainPI}, we use a vector notation for $\mathbf{\bar{D}}_{k}$ and $\mathbf{h}_{k}$, which denote the average delay and relative value for all states, respectively, following the $k$th policy iterate $\mu_{(k)}$. Likewise, $\mathbf{c}_{\mu}$ is the vector notation for the delay cost function under policy $\mu$, supplemented by the optimal minimized communication delays described by \eqref{eq:OfflineTrajOpt} and \eqref{eq:TrajMinDelta}, and $\mathbf{P}_{\mu}$ is the transition matrix under policy $\mu$.

\begin{algorithm}
\caption{Multi-chain PI to solve \eqref{eq:FinalMinSimplification2}}
\label{alg:PI}
\begin{algorithmic}[1]
\State Initialize $k =-1$, arbitrary policy $\mu_{(0)}$;

\Repeat% {$M = 1$}

\State $k \gets k+1$

\State \emph{Evaluation}: Solve for gain $\mathbf{\bar{D}}_{k}$ and relative value $\mathbf{h}_{k}$ under policy $\mu_{(k)}$ by gain-relative value optimality equations  \cite{PutermanMDP};

\State \emph{Improvement}: Find $\mu_{(k+1)}{\in}\arg\min_{\mu} \, \big\{\mathbf{P}_{\mu} \mathbf{\bar{D}}_{k} \big\}$;
choose $\mu_{(k+1)} = \mu_{(k)}$ if $\min_{\mu} \, \big\{\mathbf{P}_{\mu} \mathbf{\bar{D}}_{k}\big\}=\mathbf{P}_{\mu_{(k)}} \mathbf{\bar{D}}_{k}$;

\If {$\mu_{(k+1)} = \mu_{(k)}$}
\State Find $\mu_{(k+1)}{\in}\arg\min_{\mu} \, \big\{\mathbf{c}_{\mu}{+}\mathbf{P}_{\mu}\mathbf h_{k} \big\}$;
choose $\mu_{(k+1)}{=}\mu_{(k)}$ if $
\min_{\mu} \, \big\{\mathbf{c}_{\mu}{+}\mathbf{P}_{\mu}\mathbf h_{k} \big\}
=\mathbf{c}_{\mu_{(k)}}{+}\mathbf{P}_{\mu_{(k)}}\mathbf h_{k}$;
\EndIf

\Until $\mu_{(k+1)} = \mu_{(k)}$; return $\mu^*=\mu_{(k+1)}$.%\EndWhile

\end{algorithmic}
\label{alg:MultichainPI}
\end{algorithm}
\vspace{-3mm}
\section{Numerical Results}\label{numerical_results}

%THIS PART SEEMS REDUNDANT GIVEN THAT WE DESCRIBE ALL OF IT IN MORE DETAIL RIGHT AFTER 
%With the analysis of Section III, the original problem of \eqref{eq:OverallObj} and \eqref{eq:OverallPol} is now in the form of a SMDP and can be solved with standard dynamic programming techniques. 

%THIS PART SEEMS REDUNDANT GIVEN THAT WE DESCRIBE ALL OF IT IN MORE DETAIL RIGHT AFTER \begin{comment}
%With this approach in mind, this section shows several important features: firstly, the optimal policy for \emph{waiting states} $s \in \mathcal{S}_{\mathrm{wait}}$ is to move towards the geometric center of the line connecting GN$_1$ and GN$_2$; secondly, an inherent structure of the communication policy across a variety of values for the payload $L$; thirdly, a comparison to an intuitive heuristic, against which the optimal policy consistently outperforms; finally, we note the emerging independence of the policy of \emph{communicating states} $s \in \mathcal{S}_{\mathrm{comm}}$ when the payload value is sufficiently large (showing that it only depends on system parameters).
%\end{comment}

%To set up the simulation that provides the results of $\mu^*$ from Algorithm \ref{alg:MultichainPI}, 
We use the following system parameters, unless specified otherwise: number of states $2N{+}1{=}101$; channel bandwidth $B{=}1\mathrm{MHz}$; $1$-meter reference SNR $\gamma_{\mathrm{dB}} {=}40\mathrm{dB}$; UAV height $H{=}100\mathrm m$; GN locations $x_1{=}-400\mathrm m$, $x_2{=}400\mathrm m$; UAV speed
 $V{=}20\mathrm{m/s}$; and request arrival rate $\lambda{=}0.4$ requests/second.

We vary the data payload $L$ across a range of values and find numerically that, regardless, the optimal policy in the \emph{waiting phase} optimized with Algorithm \ref{alg:PI} is
\begin{equation}\label{eq:WaitPolicy}
 \theta^{*}(i)
{=}\left\{
\begin{array}{ll}
1, & i \in \{-N,-N+1,...,-1\} \\
0, & i = 0 \\
-1, & i \in \{1,2,...,N\}.
\end{array}
\right.
\end{equation}
In other words, in the \emph{waiting phase}
 it is optimal for the UAV to move towards the geometric center of the two GNs along the line segment connecting the two. Intuitively, the UAV can more readily service a request that is
originated equally likely from GN$_1$ or GN$_2$, if it is located in the geometric center when the request arrives, since the UAV is equally distant from both GNs, and can thus serve them equally well.

\begin{figure}[t]
\centering
\includegraphics[width=.8\linewidth]{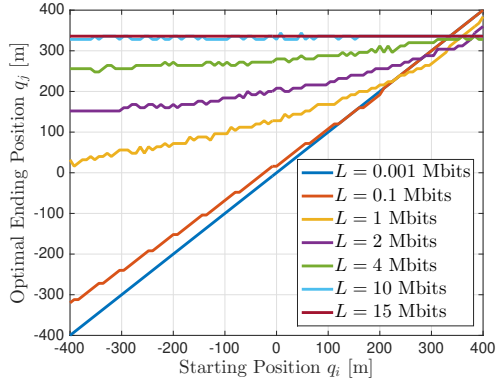}
\caption{End position in the communication phase as a function of the start position under the optimal policy, when transmitting to GN$_2$ in position $a$, for different values of the data payload. The small fluctuations are due to the discretization of the state space.}
\label{fig:OptimalReq1Policy}
\vspace{-5mm}
\end{figure}

In Fig. \ref{fig:OptimalReq1Policy}, we plot the optimal 
\emph{end position policy} $J^*(i,2)$ for different data payload values.\footnote{We omit the figure for states $(i,1) \in \mathcal{S}_{\mathrm{comm}}$, due to the inherent symmetry of the problem. Specifically, if the optimal end point $J^*(i,2) = j$ is observed, then $J^*(-i,1) = -j$ is also observed.}
We note that, for large data payload values $L$, the optimal
end position in the communication phase becomes independent of the initial position $i$ (in this case, $J^*(i,2){\approx} 336\mathrm m$, irrespective of $i$ for $L{\gg} 1$).
In fact, for large data payload $L$, the UAV hovers over the receiver for a significant amount of time during the communication phase (case $\ell_{q_i,x_r}^{(r)}{+}\ell_{x_r,q_j}^{(r)} {\leq} L$ in Theorem~\ref{th:OptTrajCase1}),
hence the final part of the trajectory from $x_r$ to the selected end position $q_j$ becomes irrespective of the actual data payload value. However, $J^*(i,2)$ does depend on other system parameters, such as the request rate $\lambda$ and UAV height $H$, as seen in Fig.~\ref{fig:EndPosAndLambda}.
Interestingly, as the request rate increases (the inter-arrival request time $1/\lambda$ decreases) the end position is closer to the geometric center (i.e., farther away from the receiver); this is because requests arrive more often, hence it is desirable for the UAV to terminate the communication phase closer to the center, in order to more readily serve future requests.

\begin{figure}[t]
\centering
\includegraphics[width=.8\linewidth]{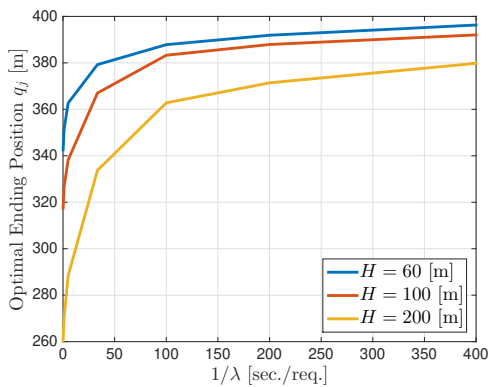}
\caption{End position for all states $(i,2)$ in the communication phase as a function $1/\lambda$, when transmitting to GN$_2$ in position $a$, for a fixed large data payload $L = 15$ Mbits (varied across UAV height).
\vspace{-7mm}}
\label{fig:EndPosAndLambda}
\end{figure}

Next, we illustrate how the optimal expected average delay $\bar{D}_{\mu}^*$, across the same set of data payload values, fares against a heuristic policy which operates as follows: in the \underline{waiting phase}, hover in the current position; in the \underline{data communication phase}, greedily minimize the delay by flying at maximum speed towards the receiver until completion. The comparison between the optimal policy $\mu^*$ and the heuristic policy is shown for the span of data payload values in Fig. \ref{fig:HeuristicComparison}. Note that the slope of the line for both the optimal and heuristic policies saturates to $[B \log_2 (1{+}\gamma/H^2)]^{-1}$. In fact, when $L{\gg}1$, the UAV spends most of the communication time hovering above the receiver (case $\ell_{q_i,x_r}^{(r)}{+}\ell_{x_r,q_j}^{(r)}{\leq}L$ in Theorem~\ref{th:OptTrajCase1}), hence
$\Delta_r^*(i,j){\approx}\frac{L}{R_{r}(x_r)}$ in \eqref{eq:FinalMinSimplification2}, yielding
\begin{equation*}
	\bar{D}_{\mu}^* \approx
	\frac{\min\limits_{\theta,J}\, \sum_{(i,r) \in \mathcal{S}_{\mathrm{comm}}}  \, \Pi_{\theta,J}(i,r)L}{\pi_{\mathrm{comm}}B \log_2 (1+\gamma/H^2)}
	=
	\frac{L}{B \log_2 (1+\gamma/H^2)}.
\end{equation*}
Overall, the heuristic scheme performs worse, roughly by $2$ seconds for large $L$. 
In fact, when hovering during the waiting phase instead of moving towards the center, the UAV incurs a larger delay to serve a request generated by the more distant GN, due to the longer distance that needs to be covered.
\vspace{-2mm}
\section{Conclusions}\label{discussion}
In this paper, we studied the trajectory optimization problem of one UAV servicing random downlink transmission requests by two GNs, to minimize the expected communication delay. We formulated  the problem as an SMDP, and
exploited the structure of the problem to simplify the trajectory design in the communication phase. We showed that the problem exhibits an interesting two-scale structure in the optimal trajectory design, and can be solved efficiently via dynamic programming. 
Numerical evaluations demonstrate consistent improvements in the delay performance over a sensible heuristic, for a variety of data payload values.

\begin{figure}[t]
\centering
\includegraphics[width=.8\linewidth]{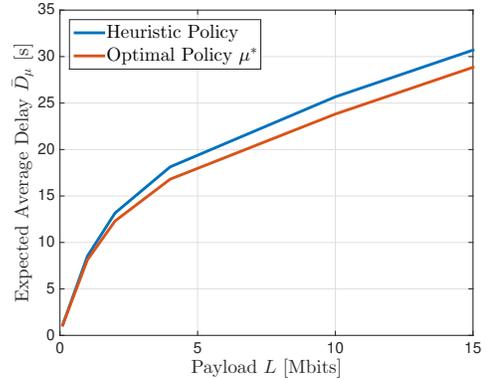}
\caption{Comparison of expected average delay $\bar{D}_{\mu}$ vs. data payload $L$ for both optimal and heuristic policy.\vspace{-5mm}}
\label{fig:HeuristicComparison}
\end{figure}

%\nm{We need a comparison with the following heuristic policy: hover until receiving a request; when request received, fly as fast as possible and as close as possible to the receiver, until completion; after completion, hover again while waiting for the next request, and so on.}
\vspace{-2mm}
\bibliographystyle{IEEEtran}
\bibliography{IEEEabrv,ref} 

\end{document}